\documentclass[letterpaper, 10 pt, conference]{ieeeconf}
\IEEEoverridecommandlockouts   
\usepackage{graphicx}  
\usepackage[tbtags]{amsmath}  
\usepackage{amssymb}   
\usepackage{microtype} 
\usepackage{cite}
\usepackage{mathtools}
\usepackage{cleveref}
\usepackage{color}
\usepackage{soul}
\usepackage{bm}
\usepackage[caption=false,font=footnotesize]{subfig}
\usepackage{siunitx}
\graphicspath{{figures/}{../figures/}}

\newtheorem{remark}{Remark}
\newtheorem{assumption}{Assumption}
\newtheorem{theorem}{Theorem}[section]

\Crefname{figure}{Figure}{Figures}
\Crefname{section}{Section}{Sections}

\crefname{section}{Sec.}{sections}
\crefname{figure}{Fig.}{figures}
\crefname{table}{Table}{tables}
\crefname{assumption}{Assumption}{Assumptions}
\crefname{remark}{Remark}{Remarks}
\crefname{theorem}{Theorem}{Theorems}

\crefformat{equation}{(#2#1#3)}
\crefrangeformat{equation}{(#3#1#4) to~(#5#2#6)}
\crefmultiformat{equation}{(#2#1#3)}%
{ and~(#2#1#3)}{, (#2#1#3)}{ and~(#2#1#3)}
\crefrangemultiformat{equation}{(#3#1#4) to~(#5#2#6)}%
{ and~(#3#1#4) to~(#5#2#6)}{, (#3#1#4) to~(#5#2#6)}{ and~(#3#1#4) to~(#5#2#6)}

\newcommand{\C}[1]{\mathcal{#1}}

\newcommand{\RealDim}[1]{\mathbb{R}^{#1}}
\newcommand{\norm}[1]{\left\lVert#1\right\rVert}

\newcommand{\state}{x}
\newcommand{\statepred}{\hat{x}_{\mathrm{RK}}}
\newcommand{\stated}{r}
\newcommand{\stateerr}{\tilde{x}}

\newcommand{\zstate}{z}
\newcommand{\zpred}{\hat{z}_{\mathrm{RK}}}

\newcommand{\act}{\eta}
\newcommand{\actr}{\act^{*}}

\newcommand{\actd}{\bar{\act}}
\newcommand{\actdfo}{\actd'}
\newcommand{\actdobs}{\actd''}
\newcommand{\acterr}{\tilde{\act}}

\newcommand{\aobs}{\hat{\act}}
\newcommand{\aobserr}{\act_e}
\newcommand{\actpred}{\aobs_{\mathrm{RK}}}

\newcommand{\sig}{u}

\newcommand{\lipf}{L_f}
\newcommand{\lipg}{L_g}
\newcommand{\lipa}{L_{\actd}}
\newcommand{\lipadot}{L_{\dot{\actd}}}

\newcommand{\liprk}{L_{\mathrm{RK}}}
\newcommand{\sampleind}{i}
\newcommand{\ts}{t_{\sampleind}}
\newcommand{\tsp}[1]{t_{\sampleind+#1}}
\newcommand{\tsm}[1]{t_{\sampleind-#1}}
\newcommand{\Ts}{T}
\newcommand{\rk}{\C{F}_{\mathrm{RK}}}
\newcommand{\istep}{h}
\newcommand{\iorder}{p}
\newcommand{\ferror}{w}
\newcommand{\boundrk}{E_{\mathrm{RK}}}
\newcommand{\delay}{\Delta}
\newcommand{\delayc}{\delay_c}
\newcommand{\delays}{\delay_s}
\newcommand{\ldelay}{\Lambda}
\newcommand{\compf}{C_f}
\newcommand{\compi}{C_0}
\newcommand{\compc}{C_{\act}}
\newcommand{\lyap}{\C{V}}
\newcommand{\eigmin}{\mathrm{EIG}_{\min}}
\newcommand{\eigmax}{\mathrm{EIG}_{\max}}
\newcommand{\gammamin}{\gamma_{\min}}
\newcommand{\gammamax}{\gamma_{\max}}
\newcommand{\lambdamin}{\lambda_{\min}}
\newcommand{\lambdamax}{\lambda_{\max}}
\newcommand{\omegamin}{\omega_{\min}}
\newcommand{\omegamax}{\omega_{\max}}

\title{\LARGE \bf Numerical Predictive Control for Delay Compensation}

\author{Xichen Shi, Michael O'Connell, and Soon-Jo Chung
\thanks{Xichen Shi, Michael O'Connell, and Soon-Jo Chung are with California Institute of Technology. \texttt{\{xshi, moc, sjchung\}@caltech.edu}.}
}

\begin{document}

\maketitle
\thispagestyle{empty}
\pagestyle{empty}

\begin{abstract}

We present a delay-compensating control method that transforms exponentially stabilizing controllers for an undelayed system into a sample-based predictive controller with numerical integration.  
Our method handles both first-order and transport delays in actuators and trades-off numerical accuracy with computation delay to guaranteed stability under hardware limitations.
Through hybrid stability analysis and numerical simulation, we demonstrate the efficacy of our method from both theoretical and simulation perspectives.


\end{abstract}

\section{Introduction}
State or control delays occur naturally in a variety of physical and cyber-physical systems. Since its introduction in 1946~\cite{tsypkin1946systems, smith1959controller, richard2003time, krstic2009delay}, time-delayed dynamics have been an active area of research and is seeing continued interest, with the popularization of vast computer networks and internet-of-things (IoT) accompanied by substantial communication lags~\cite{gao2008new,gupta2009networked}. Delay compensation techniques have also been widely used in control of power electronics~\cite{cortes2011delay,lu2017graphical} and reinforcement learning settings~\cite{schuitema2010control}.

For linear systems, delay for unstable process are often modeled as first or second order plus dead time (FOPDT or SOPDT). Classical linear feedback control can be applied and closed-loop system behavior is analyzed with transfer function approaches. It was shown that properly designed proportional-integral-derivative (PID) controllers can act as a delay compensator~\cite{visioli2006practical}. Other popular techniques include relay-based identification~\cite{padhy2006relay} and proportional-integral-proportional-derivative (PI-PD) control~\cite{majhi2000online}. For nonlinear systems, the usual consensus on the challenge of continuous delays is that the state space becomes infinite dimensional. Thus, instead of being described by ordinary differential equations (ODEs), these systems need to be modeled as functional differential equations (FDEs) or transport partial differential equations (PDEs)~\cite{richard2003time, krstic2009delay}. Accordingly, their analysis requires additional mathematical tools such as Lyapunov-Krasvoskii functionals~\cite{kharitonov2003lyapunov, mazenc2012lyapunov}. A prominent class of delay compensation methods rely on state predictions of some kind. This idea was first proposed as the Smith-predictor~\cite{smith1959controller}, and has been expanded to handle unstable processes~\cite{henson1994time}, increase robustness against uncertainties~\cite{roh1999robust}, or adapt to varying delays~\cite{bresch2009adaptive}. In theory, predictor-based methods can handle arbitrarily large delays for forward complete and strict-feedforward systems~\cite{krstic2009input}.

The FDE or PDE modeling approach has the underlying assumption that input signal is continuous in time. For control systems run on digital computers in practice, this assumption is only true when the evaluation time of the controller is much smaller compared to the transport delay of the signal. The statement is largely valid for cases considered in networked control system. However, certain real-time control applications with limited computation capacity tend to violate the continuity assumption, since controller calculation time runs at similar timescales as other delays as illustrated in~\cref{fig:sendrecq}. We take interest in the following aspects of such systems: First, the control input often corresponds to commands on actuators, which admit additional layers of control that act as a dynamic delay; Second, computation time of the controller is non-negligible and is affected by the complexity of the control algorithm; Last but not least, the discrete sampling for the control implementation poses restrictions on the stability for the continuous dynamics.

\begin{figure}[!t]
\centering
\includegraphics[width=\linewidth]{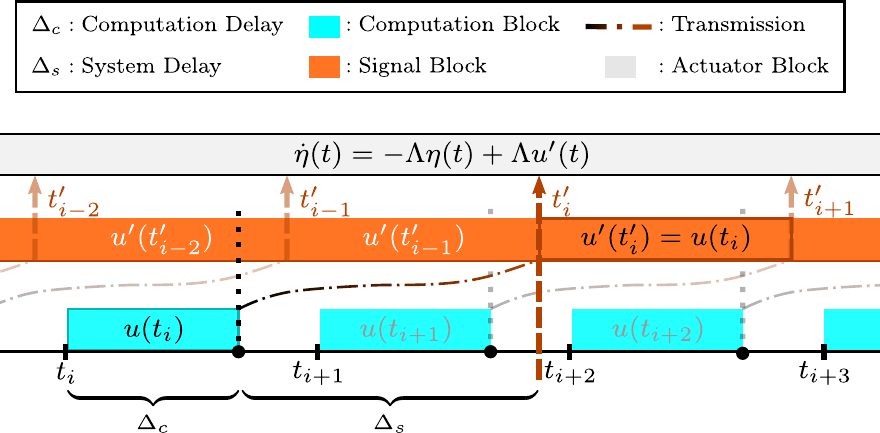}
\caption{Timeline of periodic control with computation, system, and actuator delays.  At every $t_i$, the controller begins computing a new command, $u(t_i)$, which takes $\delayc$ to calculate and an additional $\delays$ to be received and applied by the actuators.}
\label{fig:sendrecq}
\vspace{-5mm}
\end{figure}

When actuator measurements are available, it is straight forward to include actuator dynamics in the full system control design and adjust for the additional transport delay. In cases where such measurements are inaccessible, actuator observers can be constructed. This is common in applications such as multirotor control when delays exist in motor speed but output rotation may not be available~\cite{faessler2016thrust,chen2019adaptive}. Without the assumption of a continuous control signal, we resolve to use hybrid stability analysis in place of Lyapunov-Krasvoskii approach. Similar methods have been employed to show that input-to-state (ISS) stable systems inherits robustness against effects of discrete sampling or reasonable actuation delays~\cite{tabuada2007event,mazo2009input,theodosis2018self}.

\textit{Contribution \& Organization}:
In this paper we propose a periodic predictor-based controller with numerical integration or differentiation. The system in consideration includes both dynamic and transport delays. In~\cref{sec:prob}, we first introduce the undelayed, nonautonomous system of state and actuator input. Next, we make assumptions on the associated controller design. Then, we describe the sample-based FOPDT model for actuator delay. In~\cref{sec:ctrl}, we progressively augment an existing controller to compensate for dynamic and transport delays. Hybrid stability analysis is provided to study the effects of sampling and numerical methods. In~\cref{sec:analysis}, we test an example system numerically for various attributes theorized. Lastly, concluding remarks are stated in~\cref{sec:conclusion}.

\section{Problem Formulation}
\label{sec:prob}
\subsection{Notations}
We denote $\norm{x}$ as the $2$-norm for $x \in \RealDim{n}$; $\eigmin(K)$ and $\eigmax(K)$ as the minimum and maximum eigenvalues of positive definite matrix $K$ respectively. Let $[x;y;z]$ be a stack of vectors by column, and $[x, y, z]$ be one by row. We also use $I$ to represent identity matrix of appropriate size.  We define function $f : \RealDim{n} \to \RealDim{m}$ to be Lipschitz continuous on compact sets if $\forall$ compact set $\C{S} \subset \RealDim{n}$, $\exists$ constant $L$ such that $\norm{f(a) - f(b)} \leq L \norm{a - b} \ \forall a, b \in \C{S}$; function $r : \RealDim{n} \to \RealDim{m}$ is $\C{C}^k$ smooth if all of its partial derivatives up to order $k$ are continuous.
\subsection{Nonautonomous Dynamics of Trajectory Tracking}
Consider the system described by nonlinear and nonautonomous dynamics of the form
\begin{equation}
    \dot{\state} = f(\state, \act, t)
    \label{eq:dyn}
\end{equation}
where $\state \in \RealDim{n}$ is the $n$-dimensional state, and $\act \in \RealDim{m}$ is the $m$-dimensional actuator input. Given a smooth, time-prescribed, feasible reference trajectory $\stated(t)$, along with the corresponding reference control $\actr(t)$, we define the state error as $\stateerr(t) = \state(t) - \stated(t)$ corresponding dynamics
\begin{equation}
    \dot{\stateerr} = g(\stateerr, \act, t).
    \label{eq:err_dyn}
\end{equation}
$g(\stateerr, \act, t)= f(\stateerr + \stated(t), \act, t) - \dot{\stated}(t)$ is transformed from~\cref{eq:dyn}. Without loss of generality, we will focus our analysis on system~\cref{eq:err_dyn} in this paper. We also assume $\stated(t)$ is feasible for~\cref{eq:dyn} with $\actr(t)$ that guarantees $f(\stated(t), \actr(t), t) = \dot{\stated}(t)$. Therefore, along $\stated(t)$ we have
\begin{equation}
    0 = g(0, \actr(t), t)
    \label{eq:equilibrium}
\end{equation}
Additionally, we make the following assumptions:
\begin{assumption}
The function $f(\cdot)$ is Lipschitz continuous on compact sets with constant $\lipf$. The trajectory $\stated(t)$ is $\C{C}^2$ smooth with bounded derivatives. Thus it follows that $g(\cdot)$ is Lipschitz continuous on compact sets with constant $\lipg$.
\label{asm:lipschitz}
\end{assumption}
\begin{assumption}
The full state vector $\state$ is observable, the analytical form of $\stated(t)$ and its derivatives are known, but $\act$ cannot be measured directly.
\label{asm:state_measure}
\end{assumption}

Assumptions~\ref{asm:lipschitz} and \ref{asm:state_measure} are not overly restrictive. A wide class of dynamic systems possess these properties. The unavailability of measuring $\act$ is intentional, and variation of our method can compensate for delay in $\act$ without its feedback.
\subsection{Exponentially Stabilizing Control for Undelayed System}
Suppose a feedback controller of the form 
$
    \act = \actd \left(\stateerr(t), t\right)
$
has been designed, such that when applied to~\cref{eq:err_dyn}, the closed-loop system 
$
    \dot{\stateerr} = g\big(\stateerr, \actd(\stateerr, t), t\big)
$
is exponentially stable. By the Converse Lyapunov Theorem~\cite{khalil2002nonlinear}, there exists a smooth Lyapunov function $\lyap(\stateerr, t)$ such that
\begin{subequations}
\begin{gather}
    c_1 \norm{\stateerr}^2 \leq \lyap(\stateerr, t) \leq c_2 \norm{\stateerr}^2 \\
    \frac{\partial\lyap}{\partial t} + \frac{\partial\lyap}{\partial \stateerr} g(\stateerr, \actd, t)\leq -c_3 \norm{\stateerr}^2 \\
    \norm{\frac{\partial \lyap}{\partial \stateerr}} \leq c_4 \norm{\stateerr}
\end{gather}
\label{eq:converse_lyap}
\end{subequations}
Likewise, we assume smoothness of the controller function:
\begin{assumption}
The function $\actd(\cdot)$ is Lipschitz continuous on compact sets with constant $\lipa$.
\label{asm:lip_act}
\end{assumption}

We can differentiate $\actd(\stateerr, t)$ and use~\cref{eq:err_dyn} to get
\begin{equation}
    \dot{\actd}(\stateerr, \act, t) = \frac{\partial\actd}{\partial\stateerr}g(\stateerr, \act, t) + \frac{\partial \actd}{\partial t}.
    \label{eq:actdot_analytic}
\end{equation}
Based on~\cref{asm:lip_act}, it can be shown that $\dot{\actd}(\stateerr, \act, t)$ is also Lipschitz continuous on compact sets, and we define its Lipschitz constant as $\lipadot$.
\subsection{Delay in Systems with Sample-based Control}
In practice, control input $\act(t)$ lag behind the actual command signal $u(t)$ generated by a sample-based control system. We choose to describe the combined delay as a sample-based first-order plus dead time (FOPDT) model defined between sample interval $t\in [\ts', \tsp{1}')$:
\begin{equation}
    \dot{\act}(t) = -\ldelay\act(t) + \ldelay \sig'(t), \quad \sig'(t) = \sig(\ts'-\delay) \label{eq:delay_all}
\end{equation}
with $\ts' = \ts + \delay$ being the time at which actuator received the control signal computed from samples at $\ts$, and $\ldelay \succ 0$ is a diagonal matrix whose entries are rates of convergence of $\act$. The signal generated by the controller $\sig(t-\delay)$ is delayed by $\delay$ when it is received by the actuator as $\sig'(t)$. \Cref{fig:sendrecq} illustrates such process at sample time $\ts$: $\delayc$ is the computation delay, which is the time needed to compute a control signal; $\delays$ is the combined system delay in other parallel processes (e.g network latency, downstream process, etc.). We express total transport delay as $\delay = \delayc + \delays$.

\section{Delay Compensation Control}
\label{sec:ctrl}
\begin{table}[t]
    \renewcommand{\arraystretch}{1.7}
    \centering
    \caption{Summary of Control Methods}
    \begin{tabular}{ p{0.20\linewidth} p{0.70\linewidth} }
    \hline
    Baseline~\cref{eq:converse_lyap}  &  $\actd \left(\stateerr(t), t\right)$\\ 
    Actuator Delay~\cref{eq:ctrl_fo} & $\actdfo(\stateerr, \act, t) = \actd(\stateerr,t) + \ldelay^{-1}\dot{\actd}(\stateerr, \act, t)$\\
    Observer-based~\cref{eq:ctrl_fo_obs} &
    { 
      $ 
        \begin{aligned}
            \actdobs(\stateerr,\aobs,t) &= (I - \ldelay^{-1}\Gamma)\aobs + \ldelay^{-1}\Gamma\actd(\stateerr,t) \\ 
            &\quad + \ldelay^{-1}\dot{\actd}(\stateerr,\aobs, t)
        \end{aligned}
      $
    }
    \\
    Predictive~\cref{eq:ctrl_sample_pred} & $\actdobs\big(\zpred(\ts+\delay), \ts+\delay\big)$ \\
    Truncated~\cref{eq:ctrl_fo_trunc} & $\actd''_{\mathrm{FO}}(\ts+\delay) = \actd(\ts) + (\ldelay^{-1} + \delay)\frac{\actd(\ts) - \actd(\tsm{1})}{\Ts}$
    \end{tabular}
    \label{tab:ctrl_summary}
    \vspace{-5mm}
\end{table}

We first devise control that compensates for first-order dynamic delay; then we introduce a general class of predictive controllers with a numerical integration scheme to account for large transport delays. The stability of the combined method will be analyzed under discrete sampling and integration. A summary of the proposed methods is shown in~\cref{tab:ctrl_summary}.
\subsection{Derivative Compensation for First-order Delay}
Consider the case with only first-order delay, when $\actd(\stateerr,t)$ is naively applied to $\sig'=\actd(\stateerr,t)$ in~\cref{eq:delay_all}, the combined closed-loop system for actuation error $\acterr = \act - \actd$ becomes
\begin{equation}
    \dot{\acterr} = -\ldelay\acterr - \dot{\actd}(\stateerr,t), \label{eq:acterr_dyn} 
\end{equation}
which can be shown using the Comparison Lemma~\cite{khalil2002nonlinear} that
\begin{equation}
    \norm{\acterr(t)} \leq \norm{\acterr(t_0) }e^{-\lambdamin(t - t_0)} + \frac{1}{\lambdamin}\sup_{\stateerr, t}\norm{\dot{\actd}(\stateerr,t)}
\end{equation}
with $\lambdamin = \eigmin(\ldelay)$ being the minimum first-order gain of actuators. Thus actuation error $\norm{\acterr(t)}$ converges exponentially to a bounded region determined by $\norm{\dot{\actd}(\stateerr,t)}$, which is affected by the smoothness of trajectory as seen in~\cref{eq:actdot_analytic}. We propose to extend the original controller with command derivative feedback to overcome such deficiency.

\begin{theorem}
With system defined in~\cref{eq:err_dyn,eq:delay_all}, and controller $\actd(\stateerr,t)$ that satisfies~\cref{eq:converse_lyap}, the augmented controller
\begin{equation}
    \sig' =\actdfo(\stateerr, \act, t) = \actd(\stateerr,t) + \ldelay^{-1}\dot{\actd}(\stateerr, \act, t)
    \label{eq:ctrl_fo}
\end{equation}
exponentially stabilizes the closed-loop systems~\cref{eq:err_dyn,eq:acterr_dyn}.
\label{thm:ctrl_fo}
\end{theorem}

\begin{proof}
We choose a candidate Lyapunov function $\lyap_1 = \lyap + \alpha \norm{\acterr}^2$, where $\lyap$ is from~\cref{eq:converse_lyap} and $\alpha > (c_4^2\lipg^2)/(8c_3\lambdamin)$. Using~\cref{eq:converse_lyap,eq:delay_all,eq:ctrl_fo}, we differentiate $\lyap_1$ with respect to $t$ and obtain
\begin{align}
    \dot{\lyap}_1 &= \frac{\partial\lyap}{\partial t} + \frac{\partial\lyap}{\partial \stateerr} \big(g(\stateerr, \act, t) \pm g(\stateerr, \actd, t)\big) + 2\alpha \acterr^\top \dot{\acterr} \notag \\
    &\leq -c_3 \norm{\stateerr}^2 +\frac{\partial\lyap}{\partial \stateerr} \big(g(\stateerr, \act, t) - g(\stateerr, \actd, t)\big) \notag \\
    & \hspace{1.2in}+ 2\alpha \acterr^\top \big(-\ldelay\act + \ldelay \sig' - \dot{\actd}\big) \notag \\
    &\leq -c_3 \norm{\stateerr}^2 + c_4 \lipg \norm{\stateerr}\norm{\acterr}- 2\alpha \lambdamin \norm{\acterr}^2 \notag \\
    &\leq -
    \begin{bmatrix}
    \norm{\stateerr} \\
    \norm{\acterr}
    \end{bmatrix}^\top
    \underbrace{
    \begin{bmatrix}
    c_3 & -c_4\lipg /2 \\
    -c_4\lipg/2 & 2\alpha\lambdamin 
    \end{bmatrix}
    }_{K_1}
    \begin{bmatrix}
    \norm{\stateerr} \\
    \norm{\acterr}
    \end{bmatrix}
    \notag
    \\
    &\leq -c_3' \norm{\theta}^2
\label{eq:lyapa_dot}
\end{align}
with $\theta = [\stateerr; \acterr]$ as the combined error vector, symmetric matrix $K_1\succ0$ given $\alpha > (c_4^2\lipg^2)/(8c_3\lambdamin)$, and $c_3' = \eigmin(K_1)$. Furthermore, let $c_1' = \min\{c_1, \alpha\}$ and $c_2' = \max\{c_2, \alpha\}$, we can get $c_1' \norm{\theta}^2 \leq \lyap_1 \leq c_2'\norm{\theta}^2$. Thus,
\begin{equation*}
    \norm{\theta(t)} \leq \sqrt{\frac{c_2'}{c_1'}}\norm{\theta(t_0)}\exp\left(-\frac{c_3'}{2c_2'}(t-t_0)\right),
\end{equation*}
which proves $[\stateerr; \acterr]$ converges exponentially with rate $c_3'/(2c_2')$.
\end{proof}
\begin{remark}
Equivalently, if $\dot{\stateerr}$ is available through direct measurement or numerical differentiation, then
\begin{equation}
    \dot{\actd}(\stateerr,\dot{\stateerr}, t) = \frac{\partial\actd}{\partial\stateerr}\dot{\stateerr} + \frac{\partial \actd}{\partial t}
    \label{eq:actdot_numeric}
\end{equation}
and controller~\cref{eq:ctrl_fo} can be implemented without $\act$ feedback. Nevertheless rate of convergence is limited by $\lambdamin$ of the underlying actuators.
\end{remark}
\subsection{Improved Delay Compensation with Actuator Observer}
An actuator-observer is needed if we were to increase the convergence rate on $\acterr$ beyond $\ldelay$. We define $\aobs \in \RealDim{m}$ to be the estimation of $\act$, and the observer error is their difference $\aobserr = \aobs -\act$. In this work, we assume an observer with the following property is available.

\begin{assumption}
An $\act$ observer can be designed such that the closed-loop dynamics of estimation error satisfies
\begin{equation}
    \dot{\aobserr} = - \Omega(\stateerr,t) \aobserr,
    \label{eq:act_obs}
\end{equation}
where $\Omega(\stateerr, t)$ is always positive definite. We can define its minimum and maximum eigenvalues as $\omegamin = \inf_{\stateerr, t}\eigmin{\Omega(\stateerr, t)}$, $\omegamax = \sup_{\stateerr, t}\eigmax{\Omega(\stateerr, t)}$.
\label{asm:act_obs}
\end{assumption}

A trivial observer of such type is $\dot{\aobs} = -\ldelay \aobs + \ldelay \sig'$, since the first-order delay is a stable system. However, if we want to increase rate of convergence of $\acterr$, it would be favorable to have $\omegamin > \lambdamin$. With availability of measurement stated in~\cref{asm:state_measure}, a reduced-order Luenberger observer for a linear system or a contraction-based PD observer for a nonlinear system~\cite{lohmiller1998contraction} can be utilized.

The observer-based delay compensation controller that increases overall rate of convergence is stated as follows.
\begin{theorem}
With the system defined in~\cref{eq:err_dyn,eq:delay_all}, and controller $\actd(\stateerr,t)$ that satisfies~\cref{eq:converse_lyap}, the augmented controller that incorporates estimated actuator input
\begin{align}
    \sig'&=\actdobs(\stateerr,\aobs,t) \notag\\
    &= (I - \ldelay^{-1}\Gamma)\aobs + \ldelay^{-1}\Gamma\actd(\stateerr,t) + \ldelay^{-1}\dot{\actd}(\stateerr,\aobs, t)
    \label{eq:ctrl_fo_obs}
\end{align}
exponentially stabilizes the closed-loop systems~\cref{eq:err_dyn,eq:acterr_dyn} with increased rate of convergence than controller~\cref{eq:ctrl_fo}.
\label{thm:ctrl_fo_obs}
\end{theorem}
\begin{proof}
Similar to~\cref{thm:ctrl_fo}, we select a candidate Lyapunov function $\lyap_2 = \lyap + \alpha\norm{\acterr}^2 + \beta\norm{\aobserr}^2$. Taking time-derivative and substituting in~\cref{eq:converse_lyap,eq:act_obs,eq:delay_all,eq:ctrl_fo_obs}, we get the following relationship after some simplifications:
\begin{align}
    \dot{\lyap}_2 &= \dot{\lyap} + 2\alpha \acterr^\top \dot{\acterr} + 2\beta \aobserr^\top \dot{\aobserr} \notag  \\
    &\leq -
    \begin{bmatrix}
    \norm{\stateerr} \\
    \norm{\acterr} \\
    \norm{\aobserr}
    \end{bmatrix}^\top
    \underbrace{
    \begin{bmatrix}
    c_3  & -c_4 \lipg  / 2 & 0\\
    -c_4 \lipg  / 2 & 2\alpha\gammamin  & -\alpha\rho \\
    0 & -\alpha\rho  & 2\beta \omegamin 
    \end{bmatrix}
    }_{K_2}
    \begin{bmatrix}
    \norm{\stateerr} \\
    \norm{\acterr} \\
    \norm{\aobserr}
    \end{bmatrix}
    \notag \\
    &\leq - c_3''\norm{\zstate}^2
    \label{eq:lyapo_dot}
\end{align}
We define the combined error vector $\zstate = [\stateerr; \acterr; \aobserr]$, constants $\gammamin =\eigmin(\Gamma)$ and $\rho=\eigmax(\Gamma-\ldelay)$. If we choose $\alpha$ and $\beta$ such that 
\begin{align*}
    \alpha &> (c_4^2\lipg^2)/(8c_3\lambdamin) \\
    \beta &> \frac{2 c_3 \alpha^2 \rho^2}{\omegamin(8 c_3 \alpha \gammamin - c_4^2 \lipg^2)}
\end{align*}
then we can guarantee $K_2\succ0$ and define $c_3''=\eigmin(K_2)$. Letting $c_1''=\min\{c_1, \alpha, \beta\}$, $c_2''=\max\{c_2, \alpha, \beta\}$, and consequently $c_1'' \norm{\zstate}^2 \leq \lyap_2 \leq c_2''\norm{\zstate}^2$, we obtain 
\begin{equation*}
    \norm{\zstate(t)} \leq \sqrt{\frac{c_2''}{c_1''}}\norm{\zstate(t_0)}\exp\left(-\frac{c_3''}{2c_2''}(t-t_0)\right),
\end{equation*}
which proves $[\stateerr; \acterr; \aobserr]$ converges exponentially with rate $c_3''/(2c_2'')$
\end{proof}
\begin{remark}
Although the overall rate of convergence is improved with the introduction of observer~\cref{eq:act_obs}, tracking performance is now tied with estimation error $\aobserr$, which will be affected by sensor noise or model error in practice. 
\end{remark}
\begin{remark}
When setting $\Gamma = \Lambda$, \Cref{eq:ctrl_fo_obs} reduces to~\cref{eq:ctrl_fo}, and the dependence on $\aobs$ is dropped. Thus we can treat \cref{thm:ctrl_fo} as a special case of \cref{thm:ctrl_fo_obs}.
\end{remark}

\subsection{Numerical Predictive Control under Periodic Sampling}
Starting with the continuous time formulation from~\cref{thm:ctrl_fo_obs}, we propose to extend the controller with predicted future states to account for transport delays. In the literature (e.g, \cite{krstic2009delay}), predictors are often treated as continuous integration of dynamics from current state:
\begin{equation*}
    \hat{\state}(t+\delay) = \state(t) + \int_{t}^{t+\delay} f(\hat{\state}(s), \sig(s-\delay), s) ds
\end{equation*}
Instead, we consider a predictor in the form of discrete numerical integration. Our controller is activated periodically at sample times $\ts$. A general fixed step-size Runge-Kutta (RK) integration method is then used to predict state and actuator input at $\ts'=\ts+\delay$
\begin{equation}
    \begin{bmatrix}
    \statepred(\ts') \\
    \actpred(\ts')
    \end{bmatrix}
    =
    \rk\left(\state(\ts), \aobs(\ts), \ts, \delay, \istep, \iorder\right).
    \label{eq:rk_pred}
\end{equation}
We denote $\rk(\cdot)$ as the integration scheme, with accuracy of order $p$, stepsize $h$ and time horizon $\delay$. 

For ease of analysis, we vertically stack $\zstate = [\stateerr; \acterr; \aobserr]$, and rewrite $\act = \actd(\stateerr, t) + \acterr$. Then we have
\begin{equation}
    \dot{\zstate} =
    \begin{bmatrix}
    g\left(\stateerr, \act, t\right) \\
    -\ldelay \actd(\stateerr, t) - \ldelay \acterr - \dot{\actd}(\stateerr, \act, t) + \ldelay \sig \\
    -\Omega(\stateerr, t)\aobserr
    \end{bmatrix}
    = \xi(\zstate, \sig, t).
\label{eq:dyn_z}
\end{equation}
Furthermore, we limit $p \in \{1, 2, 3, 4\}$ and make the following assumption about the bound on the integration error.
\begin{assumption}
The integration error from $\ts$ to $\ts+\delay$ is bounded by $\boundrk$ as
\begin{align}
    \norm{\zpred - \zstate}_{\ts'} \leq \boundrk = \frac{M \istep^{\iorder} + \ferror}{\liprk}\left(e^{\liprk\delay} - 1\right).
    \label{eq:rk_bound}
\end{align}
$\liprk$ is the Lipschitz constant of the one-step RK function~\cite{atkinson2011numerical}; $\ferror$ is the upper bound on model error; and $M$ is a constant related to the smoothness of $\xi(\cdot)$.
\end{assumption}

Before stating the result for the predictive controller, we define the following useful quantities based on Lipschitz constants of $g(\cdot)$, $\actd(\cdot)$ and $\dot{\actd}(\cdot)$:
\begin{align}
    \mu &= \sqrt{3}\max\big\{\rho + \lipadot, \ \lambdamax\lipa + \lipadot(1+\lipa)\big\} \\
    \nu &= \sqrt{3}\max\big\{\lambdamax\lipa + (\lipg +\lipadot)(1+\lipa),  \\
    & \qquad \qquad \qquad  \lambdamax+\lipadot+\lipg, \ \omegamax+\lipadot \big\} \notag \\ 
    \nu_0 &= \sqrt{3}\max\big\{\lipg(1+\lipa), \ \gammamax+\lipg, \ \rho+\omegamax \big\}
    \label{eq:lyap_constants}
\end{align}
The numerical predictive controller under periodic sampling can be stated as follows.

\begin{theorem}
At $t=\ts$, prediction $\zpred(\ts+\delay)$ can be estimated from numerical integration with~\cref{eq:rk_pred}. The predictive controller is defined from~\cref{eq:ctrl_fo_obs} as
\begin{equation}
    \sig(\ts) = \actdobs\big(\zpred(\ts+\delay), \ts+\delay\big).
    \label{eq:ctrl_sample_pred}
\end{equation}
Suppose the sampling period satisfies
\begin{equation}
    \Ts < \frac{1}{\nu}\ln{\left[1 + \left(\frac{\nu}{\nu_0}\right) \frac{ c_3'' }{2 \alpha \mu }\right]}.
    \label{eq:sampling_cond}
\end{equation}
Then $\exists \ 0 < \delta_1 \leq \delta_2 $ such that overall system~\cref{eq:dyn_z} is exponentially stable for $\delta_1 \leq \norm{\zstate} \leq \delta_2$ under~\cref{eq:ctrl_sample_pred}.
\end{theorem}
\begin{proof}
We start from the same Lyapunov candidate $\lyap_2(z)$ as in~\cref{thm:ctrl_fo_obs}. Differentiate $\lyap_2$ with respect to time and substitute in~\cref{eq:dyn_z,eq:ctrl_sample_pred}, we get the following inequalities after simplification
\begin{equation*}
\begin{split}
    \dot{\lyap}_2 &= \dot{\lyap} + 2\alpha \acterr^\top \dot{\acterr} + 2\beta \aobserr^\top \dot{\aobserr} \\
    &\leq -c_3''\norm{z}^2 + 2 \alpha \mu \norm{z} \bigg\{\norm{\zpred(\ts') - \zstate(\ts')} \\
    & \hspace{0.9in} + \norm{\zstate(t) - \zstate(\ts')} +  (1/\sqrt{3})\norm{t-\ts'}\bigg\}
\end{split}
\end{equation*}
We can express $\zstate(t) = \zstate(\ts') + \int_{\ts'}^{t} \xi\big(\zstate(s), \sig(\ts),s\big)ds$ using~\cref{eq:dyn_z,eq:ctrl_sample_pred}. The inequality can be reduced to 
\begin{align*}
    \norm{\zstate(t) - \zstate(\ts')} &\leq \mu \norm{\zpred(\ts') - \zstate(\ts')}(t-\ts') \\
    & \qquad+ \nu_0 \norm{\zstate(\ts')}(t-\ts') + \frac{1}{2\sqrt{3}}(t-\ts')^2 \\
    & \qquad \quad + \int_{\ts'}^{t}\nu \norm{\zstate(s) - \zstate(\ts')}ds 
\end{align*}
with~\cref{eq:equilibrium}, \cref{eq:lyap_constants}, and~\cref{asm:lipschitz}. We can apply Grönwall's lemma to the above inequality; and~\cref{eq:rk_bound} to $\norm{\zpred(\ts') - \zstate(\ts')}$:
\begin{align*}
\dot{\lyap}_2 &\leq -c_3''\norm{z}^2 + 2 \alpha \mu \norm{z} \bigg\{ \boundrk \notag \\
& + \left(\frac{\nu_0}{\nu}\norm{\zstate(\ts')} + \frac{\mu}{\nu}\boundrk + \frac{1}{\sqrt{3}\nu}\right) \left( e^{\nu(t-\ts')} - 1 \right) \bigg\}
\end{align*}
Thus, for any sampling period that satisfies~\cref{eq:sampling_cond}, the following equation holds
\begin{equation}
    \Ts = \frac{1}{\nu}\ln{\left[1 + \left(\frac{\nu}{\nu_0}\right) \frac{ \phi c_3'' }{2 \alpha \mu }\right]}
    \label{eq:phi_relation}
\end{equation}
with $\phi \in (0, 1)$. We can define $\epsilon $ such that $0 < \phi < \sqrt{\phi} < \epsilon < 1$. Therefore, for any
\begin{equation}
    \delta \geq \frac{\frac{2\alpha\mu\nu_0}{c_3''}\boundrk + (\mu \boundrk + \frac{1}{\sqrt{3}})\phi}{\nu_0 (\epsilon^2 - \phi)},
    \label{eq:delta_upper}
\end{equation}
it can be shown using~\cref{eq:phi_relation} that
\begin{equation*}
    \frac{2\alpha \mu}{\epsilon c_3''}\left[\boundrk + \frac{1}{\nu} \left(\nu_0 \delta + \mu \boundrk + \frac{1}{\sqrt{3}} \right)\left(e^{\nu \Ts} - 1\right)\right] \leq \epsilon \delta
\end{equation*}
And we can state that $\forall \norm{z} \in [\epsilon\delta, \delta]$, we have $\dot{\lyap}_2 \leq -c_3''(1-\epsilon)\norm{z}^2$, and therefore
\begin{equation*}
    \norm{\zstate(t)} \leq \sqrt{\frac{c_2''}{c_1''}}\norm{\zstate(t_0)}\exp\left(-\frac{c_3''(1-\epsilon)}{2c_2''}(t-t_0)\right),
\end{equation*}
which guarantees exponential convergence with rate $c_3''(1-\epsilon)/(2c_2'')$. Setting $\delta_1 = \epsilon\delta$ and $\delta_2 = \delta$ completes the proof.
\end{proof}
\begin{remark}
From~\cref{eq:delta_upper}, it can be shown that $\epsilon\delta$ is lower bounded by $\left[\frac{2\alpha\mu\nu_0}{c_3''}\boundrk + \left(\mu \boundrk + \frac{1}{\sqrt{3}}\right)\phi\right]/\left[\nu_0(1 - \phi)\right]$.
This gives an asymptotic region within which exponential convergence is not proven sufficiently. As $\phi \to 0$, we get its continuous limit $2\alpha\mu\boundrk/c_3''$.
\label{rmk:delta_lower}
\end{remark}
\begin{remark}
The limit for sampling time in~\cref{eq:sampling_cond} is a sufficient condition that considers the worst case of which sampling error $\norm{\zstate(t)-\zstate(\ts')}$ can grow during $t \in [\ts', \tsp{1}']$. In reality, a sampling period higher than the bound can still yield reasonable stability, as seen in event-triggered controllers~\cite{tabuada2007event}.
\end{remark}

\subsection{Effects of Numerical Prediction Scheme on Delay}
\label{ssec:rk_pred}
In the case of our proposed predictive controller~\cref{eq:ctrl_sample_pred}, we postulate $\delayc$ is mainly affected by computation of predictor~\cref{eq:rk_pred} and controller $\actd''(\cdot)$. The predictor integrates $(\delayc+\delays)/\istep$ steps of target function using RK method of order $\iorder \in \{1, 2, 3, 4\}$, which requires the evaluation of target function $\iorder$ times. Thus $\delayc$ can be written as
\begin{equation*}
    \delayc = \frac{\delayc + \delays}{\istep} (\iorder \compf + \compi) + \compc,
\end{equation*}
and $\compf$, $\compc$, and $\compi$ are the respective time for evaluating $f(\cdot)$, $\act''(\cdot)$, and other related numerical operations. In turn, $\delayc$ can be solved as
\begin{equation}
    \delayc = \frac{\istep\compc + \delays(\iorder \compf + \compi)}{\istep - \iorder \compf - \compi}.
\end{equation}
It is clear that $\istep > \iorder\compf + \compi$ is required for feasible $\delayc$, and that $\istep$ cannot exceed the total delay (i.e $\istep \leq \delayc + \delays$). Furthermore, $\delayc$ needs to fit within sampling period $\Ts$. We can derive that $\istep$ has to fall within the range:
\begin{equation}
    \istep \in \left[\frac{\Ts + \delays}{\Ts - \compc}(\iorder \compf + \compi),\quad \delays + \compc + \iorder \compf + \compi\right]
    \label{eq:istep_range}
\end{equation}
For a feasible $\istep$ to exist, it follows directly from the above equation that 
$\Ts \geq \compc + \iorder \compf + \compi$. $\delay$ can thus be represented in terms of $\istep$, $\iorder$, and other pre-determined quantities:
\begin{equation}
    \delay = (\delays + \compc)\frac{\istep}{\istep - \iorder\compf - \compi}.
    \label{eq:comp_delay}
\end{equation}
Combining~\cref{eq:rk_bound,eq:comp_delay}, we obtain
\begin{equation}
    \boundrk = \frac{M \istep^{\iorder} + \ferror}{\liprk}\left(e^{\liprk \frac{\istep(\delays + \compc)}{\istep - \iorder\compf - \compi}} - 1\right),
    \label{eq:erk_combined}
\end{equation}
which admits a minimum within~\cref{eq:istep_range}. Since the prediction error and $\boundrk$ affect the overall convergence rate of the system, a choice of $\istep$ and $\iorder$ will directly affect the controller performance.

\subsection{Truncated Predictive Control with Numerical Derivative}
\label{sec:unicorn}
We can also treat $\ldelay^{-1}$ as the diagonal matrix of time constants for the actuator dynamics. In many cases, $\ldelay^{-1}$, $\delay$ and $\Ts$ are on the same small timescale, i.e $\C{O}(\ldelay^{-1}) \sim \C{O}(\delay) \sim \C{O}(T) \ll 1$. Using 1st-order RK method (Euler's method) on~\cref{eq:ctrl_fo} and applying backward difference method to $\dot{\actd}(\ts)$, we can write the RK predictive controller as
\begin{align}
    \actd''(\ts') &= \actd(\ts) + (\ldelay^{-1} + \delay )\dot{\actd}(\ts) + \C{O}(\ldelay^{-1}\delay) \notag \\
    &= \actd(\ts) + \C{O}(\Ts^2) \notag \\
    & \qquad + (\ldelay^{-1} + \delay )\left[\frac{\actd(\ts) - \actd(\tsm{1})}{\Ts} + \C{O}(\Ts) \right] \notag \\
    &= \actd(\ts) + (\ldelay^{-1} + \delay )\left[\frac{\actd(\ts) - \actd(\tsm{1})}{\Ts} \right] + \C{O}(\Ts^2) \notag
\end{align}
where for simplicity we denote $\actd''(\ts') = \actd''\big(\zpred(\ts'), \ts'\big)$, $\actd(\ts) = \actd(\stateerr(\ts),\ts)$ and $\dot{\actd}(\ts) = \dot{\actd}\big(\stateerr(\ts),\dot{\stateerr}(\ts), \ts\big)$. We can thus define first-order truncation of the predictive controller:
\begin{equation}
    \actd''_{\mathrm{FO}}(\ts') = \actd(\ts) + (\ldelay^{-1} + \delay)\frac{\actd(\ts) - \actd(\tsm{1})}{\Ts},
    \label{eq:ctrl_fo_trunc}
\end{equation}
which has a truncation error of $\C{O}(T^2)$. The truncated controller~\cref{eq:ctrl_fo_trunc} avoids the evaluation of $\dot{\actd}(\cdot)$ and in turn $g(\cdot)$. Similar to~\cref{eq:ctrl_fo}, it also avoids the need for $\aobs$ and therefore saving computation on observer as well. As will be seen in later analysis, \cref{eq:ctrl_fo_trunc} performs favorably compared to more complex methods for a certain class of systems.

\section{Numerical Analysis}
\label{sec:analysis}
In this section, we conduct numerical experiments on a delayed double integrator example running our proposed control methods described in~\cref{sec:ctrl}. 

\subsection{Example: Delayed Double Integrator}

\begin{table}[b]
    \renewcommand{\arraystretch}{1.3}
    \centering
    \caption{Baseline parameters for Delayed Double Integrator}
    \begin{tabular}{c c c c c c c}
    $b$ & $\lambda$ & $k_1$ & $k_2$ & $\compf$ & $\compi$ & $\compc$  \\
    \hline
     $1.0$ & $5.0$ & $1.0$ & $2.0$ & $0.005$ & $0.0$ & $0.025$  \\
    \hline
    \end{tabular}
    \label{tab:ddi_params}
\end{table}

We consider trajectory tracking for simple double integrator dynamics, with delayed force actuation:
\begin{equation}
    \dot{\state}_1 = \state_2, \quad \dot{\state}_2 = b  \act_1, \quad \dot{\act}_1 = -\lambda\act_1 + \lambda\sig(t-\delay).
    \label{eq:eg_ddi}
\end{equation}
Let the scalar states $\state_1$ and $\state_2$ denote position and velocity respectively. $b$ is a known scalar actuation multiplier, and $\act_1$ is the actuator input with first-order delay $\lambda$. The goal is to track $\state_1 \to \stated(t)$ and $\state_2 \to \dot{\stated}(t)$. The error dynamics are
\begin{equation}
    \dot{\stateerr}_1 = \stateerr_2,  \hspace{0.75in}  \dot{\stateerr}_2 = b \act_1 - \ddot{\stated}(t).
\end{equation}
We use a baseline feedback linearizing controller of the form:
\begin{equation}
    \actd_1(\stateerr, t) = b^{-1}\Big(\ddot{\stated}(t) - k_1 \stateerr_1 - k_2 \stateerr_2 \Big) 
\end{equation}
which can be proven to exponentially stabilize the undelayed system. Although the base dynamics are relatively simple, the addition of an aggressive trajectory $\stated(t)$, large delays, and discrete sampling will pose difficulties for the baseline controller. We will use this example to study different effects on overall performance from various components of our proposed methods. \Cref{tab:ddi_params} lists related parameters for the system that are set or calculated.

\subsection{Effects of Computation Delay on Control Performance}
\begin{figure}[!t]
\centering{
	\subfloat[Sample period fixed at $\Ts= \SI{0.1}{s}$, $\delays \in \{0.2, 0.3\}\si{s}$ across rows, and $\ferror \in \{0.0, 0.2, 0.5\}$ across columns.]{
		\includegraphics[width=0.9\linewidth]{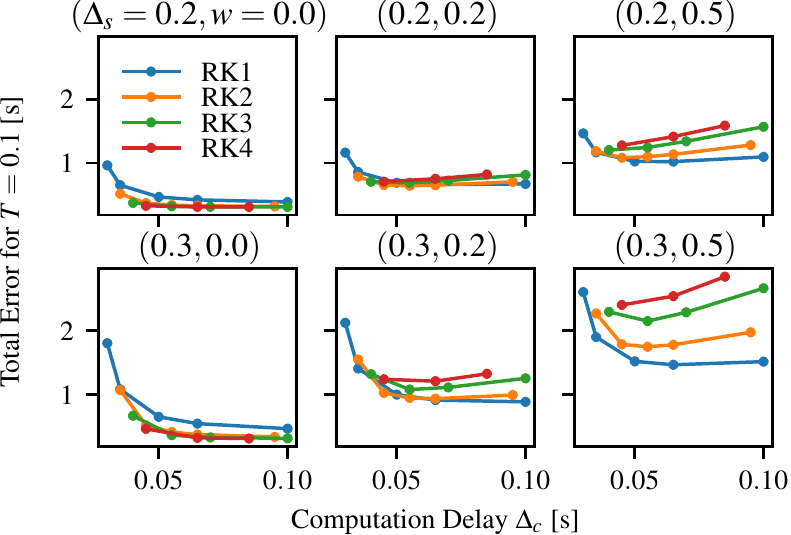}
		\label{fig:rk_trend_fixed_T}
	}
	\\
	\subfloat[Varying sample period $\Ts = \delayc$, with $\ferror \in \{0.0, 0.2, 0.5\}$]{
		\includegraphics[width=0.9\linewidth]{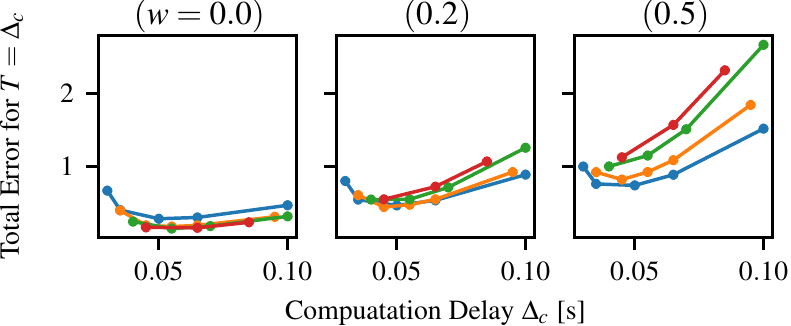}
		\label{fig:rk_trend_asap}
	}

}
\caption{Theoretical Total error bound vs. computation delay $\delayc$ for different system delay $\delays$ and model error $\ferror$. (Top) Fixed control period $\Ts = 0.1\si{s}$. (Bottom) Variable control period $\Ts = \delayc$.}
\label{fig:rk_trend_theory}
\end{figure}

\begin{figure}[!t]
\centering
\includegraphics[width=0.9\linewidth]{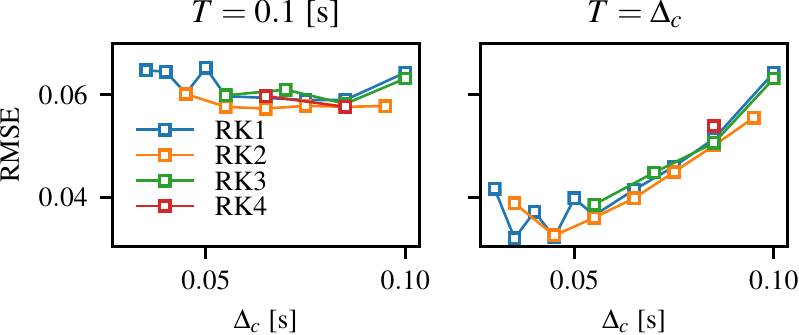}
\caption{Simulation tracking RMSEs for different integration schemes and step sizes. (Left) fixed sample period $\Ts = \SI{0.1}{s}$. (Right) variable sample period $\Ts = \delayc$.}
\label{fig:fixed-T-vs-T=DeltaC}
\vspace{-3mm}
\end{figure}

From the Lyapunov analysis in~\cref{thm:ctrl_fo_obs}, together with~\cref{eq:erk_combined}, we can predict trade-offs between integration schemes ($\iorder$ and $\istep$) and system stability by plotting the Lyapunov derivative error bound. In~\cref{fig:rk_trend_fixed_T}, we fix the sample period $\Ts = \SI{0.1}{s}$ and examine variations of system delay $\delays$ and model error $\ferror$. The total error decreases with increasing $\delayc$ for small $\ferror$, implying the benefit of maximizing integration accuracy as long as $\delayc \leq \Ts$. On the other hand, with higher $\ferror$, we observe a reversal in trend, where increasing numerical complexity no longer decreases total error, and computation delay $\delayc$ should be minimized for better result. In~\cref{fig:rk_trend_asap}, we adapted the sampling period to the computation delay $\Ts = \delayc$. These results show that faster computation is strongly favored. Moreover, Euler's method outperforms other higher order RK's when $\ferror$ is high. Results from numerical simulation corroborates our conjecture as shown in~\cref{fig:fixed-T-vs-T=DeltaC}. We observe similar $\delayc$ versus steady state tracking root-mean-square-error (RMSE) patterns when compared to~\cref{fig:rk_trend_theory}.

\subsection{Control Performance Benchmarks}

We conduct comparisons of our proposed controllers, $\actd''(\cdot)$ and $\actd''_{\mathrm{FO}}(\cdot)$ (given in \cref{eq:ctrl_sample_pred,eq:ctrl_fo_trunc}, respectively) to the baseline controller $\actd(\cdot)$ and a reasonably tuned linear PD controller. The trajectory considered is a sine function of varying frequency. For each test, we let the system run for a horizon of $\SI{20}{s}$ and then measure the steady-state RMSE. In~\cref{fig:sys-delay-rmse}, we simulate the system with different system delays, $\delays$. Throughout the test, the predictive controller $\actd''(\cdot)$ maintains a low level of RMSE, even though it is computationally more complex and has a higher $\delayc$ than the others. As discussed in \cref{sec:unicorn}, the truncated control $\actd''_{\mathrm{FO}}(\cdot)$ is a first order approximation of the full predictive control scheme. It performs well for $\delays < \SI{0.4}{s}$ but fails for larger values. The PD control maintains stability for most of the range, but has worse RMSE than $\actd''(\cdot)$. Not surprisingly, the naively applied baseline control $\actd(\cdot)$ takes high error and becomes unstable even for moderate $\delays$. \Cref{fig:freq-rmse}) shows almost identical rankings in RMSEs. It's interesting to note that $\actd''_{\mathrm{FO}}(\cdot)$ outperforms its more sophisticated counterpart $\actd''(\cdot)$ for small delay, likely due to significant reduction in computation cost.

To characterize the performance of $\actd''_{\mathrm{FO}}(\cdot)$ on different types of delay, we put it through varying combinations of $\lambda$ and $\delay$. \Cref{fig:uni-delay-ratio} shows that the truncated controller is delay-type agnostic when $\delay+1/\lambda$ is moderate, and has trouble dealing with larger $\delay$. This is expected since the assumption of $\C{O}(1/\lambda) \sim \C{O}(\delay)$ breaks down for large $\delay$.

\begin{figure}[t]
    \centering{
    	\subfloat[RMSE vs. system delay $\delays$.]{
    		\includegraphics[width=0.45\linewidth]{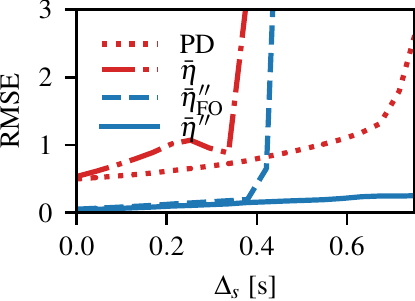}
    		\label{fig:sys-delay-rmse}
    	}
    	\hspace{0.5pt}
    	\subfloat[RMSE vs. $\stated(t)$ frequency.]{
    		\includegraphics[trim=0 -3 0 0,clip,width=0.45\linewidth]{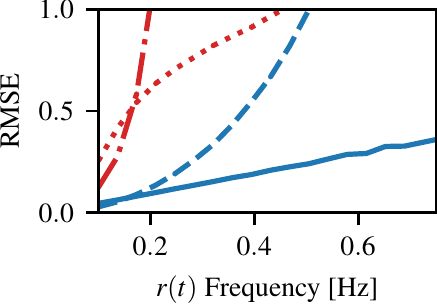}
    		\label{fig:freq-rmse}
    	}
    }
    \caption{Comparisons of PD, baseline $\actd(\cdot)$, truncated $\actd''_{\mathrm{FO}}(\cdot)$ and full predictive control $\actd''(\cdot)$. RMSE is steady-state root-mean-square-error.} 
    \vspace{-3mm}
\end{figure}

\begin{figure}
    \centering
    \includegraphics[width=\linewidth]{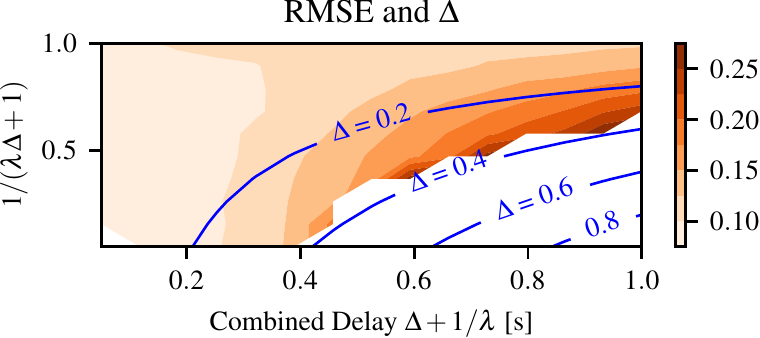}
    \caption{Contour of RMSE and transport delay $\delay$ for truncated predictive control $\actd''_{\mathrm{FO}}(\cdot)$~\cref{eq:ctrl_fo_trunc}. The horizontal axis is the combined delay ($\delay + 1/\lambda$) in seconds. The vertical axis is its ratio of first-order delay $1/(\lambda \delay + 1)$.}
    \label{fig:uni-delay-ratio}
    \vspace{-3mm}
\end{figure}


\addtolength{\textheight}{-5.5cm}

\section{Conclusion}
\label{sec:conclusion}
We proposed a control augmentation strategy that transforms exponentially-stabilizing controllers for an undelayed system to a class of sample-based, predictive controllers with numerical integration.  The predictive controllers exponentially stabilize the corresponding sample-based system with FOPDT delay and can change computation complexity under different conditions. We performed hybrid stability analysis on the overall system, which provided insights on how discrete time features such as sampling period and integration step and order affect output stability. We demonstrated the efficacy of our methods through numerical analysis of our theoretical bounds and simulations of a delayed double integrator system using our proposed control methods.  
Our analysis demonstrated the often overlooked importance of computation delay in control design. In conclusion, our predictive controller and its truncated variants provide an easily applicable improvement for discrete control tasks in different computation, network, and dynamic environments.







\bibliographystyle{IEEEtran}
\bibliography{IEEEabrv,ref}

\begin{thebibliography}{10}
\providecommand{\url}[1]{#1}
\csname url@rmstyle\endcsname
\providecommand{\newblock}{\relax}
\providecommand{\bibinfo}[2]{#2}
\providecommand\BIBentrySTDinterwordspacing{\spaceskip=0pt\relax}
\providecommand\BIBentryALTinterwordstretchfactor{4}
\providecommand\BIBentryALTinterwordspacing{\spaceskip=\fontdimen2\font plus
\BIBentryALTinterwordstretchfactor\fontdimen3\font minus
  \fontdimen4\font\relax}
\providecommand\BIBforeignlanguage[2]{{%
\expandafter\ifx\csname l@#1\endcsname\relax
\typeout{** WARNING: IEEEtran.bst: No hyphenation pattern has been}%
\typeout{** loaded for the language `#1'. Using the pattern for}%
\typeout{** the default language instead.}%
\else
\language=\csname l@#1\endcsname
\fi
#2}}

\bibitem{tsypkin1946systems}
Y.~Z. Tsypkin, ``The systems with delayed feedback,'' \emph{Avtomathika i
  Telemech}, vol.~7, pp. 107--129, 1946.

\bibitem{smith1959controller}
O.~J. Smith, ``A controller to overcome dead time,'' \emph{ISA J.}, vol.~6, pp.
  28--33, 1959.

\bibitem{richard2003time}
J.-P. Richard, ``Time-delay systems: an overview of some recent advances and
  open problems,'' \emph{Automatica}, vol.~39, no.~10, pp. 1667--1694, 2003.

\bibitem{krstic2009delay}
M.~Krstic, \emph{Delay compensation for Nonlinear, Adaptive, and PDE
  Systems}.\hskip 1em plus 0.5em minus 0.4em\relax Springer, 2009.

\bibitem{gao2008new}
H.~Gao, T.~Chen, and J.~Lam, ``A new delay system approach to network-based
  control,'' \emph{Automatica}, vol.~44, no.~1, pp. 39--52, 2008.

\bibitem{gupta2009networked}
R.~A. Gupta and M.-Y. Chow, ``Networked control system: Overview and research
  trends,'' \emph{IEEE transactions on industrial electronics}, vol.~57, no.~7,
  pp. 2527--2535, 2009.

\bibitem{cortes2011delay}
P.~Cortes, J.~Rodriguez, C.~Silva, and A.~Flores, ``Delay compensation in model
  predictive current control of a three-phase inverter,'' \emph{IEEE
  Transactions on Industrial Electronics}, vol.~59, no.~2, pp. 1323--1325,
  2011.

\bibitem{lu2017graphical}
M.~Lu, X.~Wang, P.~C. Loh, F.~Blaabjerg, and T.~Dragicevic, ``Graphical
  evaluation of time-delay compensation techniques for digitally controlled
  converters,'' \emph{IEEE Transactions on Power Electronics}, vol.~33, no.~3,
  pp. 2601--2614, 2017.

\bibitem{schuitema2010control}
E.~Schuitema, L.~Bu{\c{s}}oniu, R.~Babu{\v{s}}ka, and P.~Jonker, ``Control
  delay in reinforcement learning for real-time dynamic systems: a memoryless
  approach,'' in \emph{2010 IEEE/RSJ International Conference on Intelligent
  Robots and Systems}, 2010, pp. 3226--3231.

\bibitem{visioli2006practical}
A.~Visioli, \emph{Practical PID control}.\hskip 1em plus 0.5em minus
  0.4em\relax Springer Science \& Business Media, 2006.

\bibitem{padhy2006relay}
P.~K. Padhy and S.~Majhi, ``Relay based pi--pd design for stable and unstable
  fopdt processes,'' \emph{Computers \& chemical engineering}, vol.~30, no.~5,
  pp. 790--796, 2006.

\bibitem{majhi2000online}
S.~Majhi and D.~Atherton, ``Online tuning of controllers for an unstable fopdt
  process,'' \emph{IEE Proceedings-Control Theory and Applications}, vol. 147,
  no.~4, pp. 421--427, 2000.

\bibitem{kharitonov2003lyapunov}
V.~L. Kharitonov and A.~P. Zhabko, ``Lyapunov--krasovskii approach to the
  robust stability analysis of time-delay systems,'' \emph{Automatica},
  vol.~39, no.~1, pp. 15--20, 2003.

\bibitem{mazenc2012lyapunov}
F.~Mazenc, S.-I. Niculescu, and M.~Krstic, ``Lyapunov--krasovskii functionals
  and application to input delay compensation for linear time-invariant
  systems,'' \emph{Automatica}, vol.~48, no.~7, pp. 1317--1323, 2012.

\bibitem{henson1994time}
M.~A. Henson and D.~E. Seborg, ``Time delay compensation for nonlinear
  processes,'' \emph{Industrial \& engineering chemistry research}, vol.~33,
  no.~6, pp. 1493--1500, 1994.

\bibitem{roh1999robust}
Y.-H. Roh and J.-H. Oh, ``Robust stabilization of uncertain input-delay systems
  by sliding mode control with delay compensation,'' \emph{Automatica},
  vol.~35, no.~11, pp. 1861--1865, 1999.

\bibitem{bresch2009adaptive}
D.~Bresch-Pietri and M.~Krstic, ``Adaptive trajectory tracking despite unknown
  input delay and plant parameters,'' \emph{Automatica}, vol.~45, no.~9, pp.
  2074--2081, 2009.

\bibitem{krstic2009input}
M.~Krstic, ``Input delay compensation for forward complete and
  strict-feedforward nonlinear systems,'' \emph{IEEE Transactions on Automatic
  Control}, vol.~55, no.~2, pp. 287--303, 2009.

\bibitem{faessler2016thrust}
M.~Faessler, D.~Falanga, and D.~Scaramuzza, ``Thrust mixing, saturation, and
  body-rate control for accurate aggressive quadrotor flight,'' \emph{IEEE
  Robotics and Automation Letters}, vol.~2, no.~2, pp. 476--482, 2016.

\bibitem{chen2019adaptive}
Y.~Chen and N.~O. P{\'e}rez-Arancibia, ``Adaptive control of aerobatic
  quadrotor maneuvers in the presence of propeller-aerodynamic-coefficient and
  torque-latency time-variations,'' in \emph{2019 International Conference on
  Robotics and Automation}, 2019, pp. 6447--6453.

\bibitem{tabuada2007event}
P.~Tabuada, ``Event-triggered real-time scheduling of stabilizing control
  tasks,'' \emph{IEEE Transactions on Automatic Control}, vol.~52, no.~9, pp.
  1680--1685, 2007.

\bibitem{mazo2009input}
M.~Mazo and P.~Tabuada, ``Input-to-state stability of self-triggered control
  systems,'' in \emph{Proc. 48h IEEE Conference on Decision and Control}, 2009,
  pp. 928--933.

\bibitem{theodosis2018self}
D.~Theodosis and D.~V. Dimarogonas, ``Self-triggered control under actuator
  delays,'' in \emph{IEEE Conference on Decision and Control}, 2018, pp.
  1524--1529.

\bibitem{khalil2002nonlinear}
H.~Khalil, \emph{Nonlinear Systems}.\hskip 1em plus 0.5em minus 0.4em\relax
  Prentice Hall, 2002.

\bibitem{lohmiller1998contraction}
W.~Lohmiller and J.-J.~E. Slotine, ``On contraction analysis for non-linear
  systems,'' \emph{Automatica}, vol.~34, no.~6, pp. 683--696, 1998.

\bibitem{atkinson2011numerical}
K.~Atkinson, W.~Han, and D.~E. Stewart, \emph{Numerical solution of ordinary
  differential equations}.\hskip 1em plus 0.5em minus 0.4em\relax John Wiley \&
  Sons, 2011, vol. 108.

\end{thebibliography}

\end{document}